\documentclass[aps,pra,superscriptaddress,twocolumn]{revtex4}

\usepackage{graphicx,graphics}
\usepackage{bbm,amsfonts,amsmath,amsthm,mathrsfs}
\usepackage{hyperref}

\usepackage{color}

\newcommand{\cC}{\mathcal{C}}
\newcommand{\cD}{\mathcal{D}}
\newcommand{\cE}{\mathcal{E}}
\newcommand{\cG}{\mathcal{G}}
\newcommand{\cH}{\mathcal{H}}
\newcommand{\cQ}{\mathcal{Q}}
\newcommand{\cP}{\mathcal{P}}
\newcommand{\cR}{\mathcal{R}}
\newcommand{\cU}{\mathcal{U}}
\newcommand{\sB}{\mathscr{B}}
\newcommand{\sS}{\mathscr{S}}

\newcommand{\ket}[1]{{\left|{#1}\right\rangle}}
\newcommand{\bra}[1]{{\left\langle{#1}\right|}}
\newcommand{\tr}{\mathrm{Tr}}
\newcommand{\trb}{\mathrm{Tr}_\mathrm{B}}
\newcommand{\upe}{\mathrm{e}}
\newcommand{\upd}{\mathrm{d}}
\newcommand{\id}{\mathbbm{1}}
\newcommand{\idB}{\id_{\mathrm{B}}}
\newcommand{\idS}{\id_{\mathrm{S}}}
\newcommand{\dB}{d_\mathrm{B}}
\newcommand{\Pc}{P_\cC}
\newcommand{\Pe}{\cP_\cE}
\newcommand{\Rs}{\cR_\mathrm{S}}
\newcommand{\rhoS}{\rho_\mathrm{S}}
\newcommand{\HS}{\cH_\mathrm{S}}
\newcommand{\HB}{\cH_\mathrm{B}}
\newcommand{\HSB}{H_\mathrm{SB}}
\newcommand{\Et}{\widetilde{E}}

\newcommand{\Pt}{P_{\widetilde\cE,\rho}}
\newcommand{\avg}[1]{{\left\langle#1\right\rangle}}
\newcommand{\ddA}{\delta^2\!\!A}
\newcommand{\dA}{\delta\!A}

\newtheorem{theorem}{Theorem}
\newtheorem{corollary}[theorem]{Corollary}
\newtheorem{lemma}[theorem]{Lemma}

\begin{document}

\title{Open-System Quantum Error Correction}

\author{Yink Loong Len}
\affiliation{Centre for Quantum Technologies, National University of
  Singapore, 3 Science Drive 2, Singapore 117543, Singapore} 

\author{Hui Khoon Ng}
  \affiliation{Yale-NUS College, 16 College Avenue West, Singapore 138527,
  Singapore} 
\affiliation{Centre for Quantum Technologies, National University of
  Singapore, 3 Science Drive 2, Singapore 117543, Singapore} 
\affiliation{MajuLab, CNRS-UCA-SU-NUS-NTU International Joint Research Unit, Singapore.} 
	
\date{\today}

\begin{abstract}
We study the performance of quantum error correction (QEC) on a system undergoing open-system (OS) dynamics. The noise on the system originates from a joint quantum channel on the system-bath composite, a framework that includes and interpolates between the commonly used system-only quantum noise channel model and the system-bath Hamiltonian noise model. We derive the perfect OSQEC conditions, with QEC recovery only on the system and not the inaccessible bath. When the noise is only approximately correctable, the generic case of interest, we quantify the performance of OSQEC using worst-case fidelity. We find that the leading deviation from unit fidelity after recovery is quadratic in the uncorrectable part, a result reminiscent of past work on approximate QEC for system-only noise, although the approach here requires the use of different techniques than in past work.
\end{abstract}


\begin{widetext}
\maketitle    
\end{widetext}

\section{Introduction}

To successfully implement quantum information processing (QIP) tasks, the adversarial effects of noise on the quantum system must mitigated. Quantum error correction (QEC) is a general method for active removal of noise (see, for example, Ref.~\cite{Lidar+Brun:13}). One encodes and stores information in a part of the system Hilbert space---the code space---chosen based on the nature of the noise so that an operation can be applied to recover, with a high probability of success, the stored information despite errors caused by the noise. 

The choice of an appropriate code space and an assessment of its efficacy requires knowledge of the noise. The standard description of noise, for the purpose of QIP, falls into two main types: the system-only quantum channel model and the system-bath Hamiltonian noise model. In the quantum channel model, the noise acts as a completely-positive (CP), trace-preserving (TP) map on the system. In the Hamiltonian noise model, unitary dynamics occur for the joint system-bath composite generated by a microscopic Hamiltonian, and the resulting system state is obtained by tracing out the bath degrees of freedom. 

Much of the work on QEC is based on the quantum channel model. For instance, the conditions for existence of QEC codes are phrased in terms of the Kraus operators of the quantum noise channel acting on the system \cite{Knill+Laflamme:97,Kribs+et.al.:05, Kribs+et.al.:06,HK+Prabha:10, HK+Prabha:12}. The bulk of the literature on fault-tolerant quantum computation also deals only with the quantum channel model (see, for example, Refs. \cite{Knill+et.al.:98, Aharonov+Ben-Or:08, Aliferis+et.al:06, Stephens:13, Fowler:12, Raussendorf+Harrington:07, Criger+Terhal:16}). The Hamiltonian noise model is usually employed when discussing non-Markovian dynamics of the system arising from the joint evolution with the bath. Examples include the study of continuous-time QEC \cite{Oreshkov+Brun:07}, the analysis of fault tolerance in the non-Markovian situation \cite{Terhal+Burkard:05, Aliferis+et.al:06, Aharonov+et.al:06, Ng+Preskill:09, Novais+et.al:08(1)}, and the study of entanglement revival or information backflow in a system coupled to a bath (see the review article \cite{Vega+Alonso:17} and references therein).

The quantum channel description is usually assumed when the noise on the system alone is directly characterized, for example, through the use of process tomography. The Hamiltonian noise model, as it explicitly involves the (inaccessible) bath, is used when one has a good understanding of the underlying physical processes that govern the system-bath dynamics.
The two noise models are not unrelated, of course. In many physical scenarios, the quantum channel model arises as an approximation of the Hamiltonian noise model \cite{Bruer+Petruccione:02}: When the Born-Markov approximation is appropriate, a (system-only) Lindblad master equation well describes the dynamics of the system, and the quantum channel model is but a finite time-step integration of the continuous-time Lindblad equation.

The quantum channel model is easier to work with as it refers only to the system. It is, however, inadequate in capturing non-Markovian features observable in experiments today \cite{Bernades+et.al:16, Souza+et.al:13, Chiuri+et.al:12, Xu+et.al:10, Liu+et.al:11}, and much studied in theory, often using the Hamiltonian noise model (see review articles, Refs.~\cite{Rivas+et.al:14} and \cite{Breuer+et.al:16}). Non-Markovian features in the noise can have significant consequences on the study of QEC. In fault-tolerant QEC, for example, non-Markovian noise requires, firstly, new analysis techniques in dealing with such a noise model \cite{Terhal+Burkard:05, Novais+et.al:08(1)}, and, secondly, leads to considerably worse noise threshold estimates \cite{Terhal+Burkard:05, Aliferis+et.al:06, Aharonov+et.al:06, Ng+Preskill:09}. One can expect bounds on the performance of QEC for non-Markovian noise to involve quantities that scale with the size of the bath. Intuitively, coupling to a larger bath means more ways for errors on the system to occur. However, in practice, the dependence on bath size often arises, not because the effect of the noise has some inherent scaling with the bath size, but due to proof techniques, e.g., bounding the norm of a system-bath operator appearing in the noise description. Such bounds can become unhelpful when the noise model involves a large bath, as is usually the case for the unitary dynamics of the Hamiltonian noise model to be a good physical description.

In this work, we consider a noise model that has the best of both worlds. Here, the noise is described as a \emph{joint} CPTP map (or quantum channel) acting on the system and a portion of the bath---we refer to it as the \emph{small bath}---sufficiently closely coupled to the system to have a significant non-Markovian effect on it. This mimics the typical physical situation, where the system is more closely coupled to only a few bath degrees of freedom (e.g., due to spatial proximity) and it is only these degrees of freedom that play a role of the ``memory" for the system, leading to non-Markovian effects. The bulk of the bath degrees of freedom provides only a large, dissipative, no-memory, bath that can be used to justify the Markovian approximation on the system and small-bath. Such a noise model includes and smoothly interpolates between the two standard noise models: The quantum channel model is one where the small bath is trivial, so that the ``joint" channel acts only on the system; the Hamiltonian noise model is one where the joint CPTP channel is a unitary map, when the full bath is the ``small" bath.

We want to quantify the performance of QEC under such a joint system--small-bath noise model. The information is encoded in the system, subjected to the joint system--small-bath noise, and then a recovery operation is applied. The recovery operation for error correction is applied only on the system, since the bath, by its very definition, is inaccessible and uncontrollable. We refer to this situation of QEC as ``open-system QEC" (OSQEC). That QEC works in the perfect case is perhaps not surprising---one might argue that this follows by linearity of the standard QEC conditions; in Sec.~\ref{sec:PerfQEC}, we make this precise by deriving the extension of standard QEC conditions to the case of OSQEC, and point out how it differs from what is known as \emph{operator} QEC (OQEC) \cite{Kribs+et.al.:05, Kribs+et.al.:06}. What is less obvious is how the loss in fidelity in the case of \emph{approximate} QEC (AQEC)---the practically relevant situation---depends on the open-system noise properties. In this sense, our work extends previous work on AQEC of Refs.~\cite{HK+Prabha:10, HK+Prabha:12} to this system--small-bath noise model. Section~\ref{sec:AQEC} looks at the performance of the code when the noise is approximately correctable. Before we begin, however, we first lay out the details of our noise model (Sec.~\ref{sec:NoiseModel}), and then define the basic notions of QEC, and, in particular, the new situation of OSQEC (Sec.~\ref{sec:Code}). We conclude in Sec.~\ref{sec:Conc}.

\section{The noise model}\label{sec:NoiseModel}
A physical system S, with Hilbert space denoted as $\HS$, sits within a bath B, with Hilbert space $\HB$. The system could be one intended for use in quantum information processing, and is fully controllable. In contrast, the bath is, by definition, uncontrollable and one cannot directly access its microscopic degrees of freedom. The full system-bath Hilbert space is denoted as $\cH\equiv\HS\otimes\HB$.

Noise on the system arises from interaction with the bath. We describe this noise by a 
\emph{quantum channel} acting jointly on the system and bath, i.e., a linear, completely positive (CP) and trace-preserving (TP) map $\cE:\sB(\cH)\longrightarrow\sB(\cH)$, where $\sB(\mathcal{V})$ is the set of bounded operators on a vector space $\mathcal{V}$. $\cE$ can be specified by giving a  (nonunique) set of Kraus operators $\{E_a\}_{a=1}^N$ so that $\cE$ acts as $\cE(\,\cdot\,)=\sum_{a=1}^N E_a(\,\cdot\,) E_a^\dagger$, a structure that assures the CP property. $\cE$ is TP if $\sum_{a=1}^N E_a^\dagger E_a=\id$. We write $\cE\sim\{E_a\}$ to indicate the quantum channel and an associated set of Kraus operators.

Our description of the noise as a joint system-bath CPTP map arises naturally in many physical scenarios. In a typical experiment, the system of interest is weakly coupled to the environment---the system-environment split is useful only if their mutual coupling is weak. However, invariably, there are a few environmental degrees of freedom more strongly coupled to the system than the rest of the environment, e.g., due to spatial proximity or similarity in characteristic frequencies. These more closely coupled bath degrees of freedom are responsible for non-Markovian noise or memory effects on the system. The rest of the environment forms merely an information ``sink": Any information that flows from the system into it never returns, at least not on time scales relevant for any experiment on the system. We refer to the few more-closely-coupled environmental degrees of freedom as the ``small bath", and the rest of the environment as the large bath. This identification of a large, dissipative bath permits the application of the Born-Markov approximation, where the system and small-bath undergo joint (continuous-time) dynamics according to a Lindblad master equation. Written for discrete time-steps (e.g., the time between system gate operations, or consecutive QEC cycles), the Lindbladian evolution translates into a CPTP description for propagating the system and small-bath state from one time-step to the next. The bath B of the previous paragraphs of this section is then the small bath. A single spin (system) interacting with its immediate neighbouring spins (small bath), which are themselves coupled to a larger circle of spins (large environment), or an optomechanical oscillator (system) in a lossy cavity (the near-resonant cavity modes are the small bath; the outside electromagnetic field modes form the large bath), are some examples where such a noise model would be appropriate \cite{Maze+et.al:08, Blatt+et.al:15, Muller+et.al:15, Lisenfeld+et.al:15}. 

The system-bath CPTP map description captures much more interesting and complicated dynamics of the system than the standard system-only quantum channel model. The action of $\cE$ on the system, except in special circumstances, cannot be described by a system-only CPTP map. This is the problem of (the lack of) a ``reduced system-only description" discussed in the 1990s \cite{Pechukas:94, Alicki95, Pechukas95} (if not earlier). It is known that a CPTP system-only description exists if the input system-bath state has zero discord \cite{Rodriguez+et.al:08}. More generally, one cannot even define a sensible \emph{map} on the system alone that captures the action of $\cE$. That no such system-only map exists is obvious: Two system-bath states $\rho$ and $\sigma$ satisfying $\trb\{\rho\}=\trb\{\sigma\}$, i.e., the same system-only state, will generally map to states $\rho'$ and $\sigma'$ under the action of $\cE$ such that $\trb\{\rho'\}\neq \trb\{\sigma'\}$, so one cannot possibly define a system-only map that takes in a system-only state and reproduces the correct post-$\cE$ system state \cite{note1}. Put another way, what happens to the system state under $\cE$ depends on the associated bath state, and one cannot predict the output system-only state without knowledge of the bath state and its correlations with the system. Of course, if $\cE$ has a special structure, e.g., $\cE=\cE_\mathrm{S}\otimes\cE_\mathrm{B}$, so that its action on the system and the bath are independent, one might be able to define a system-only map; generally, this would not be possible. In this sense, our system-bath CPTP $\cE$ is a \emph{bona fide} generalization of the usual system-only CPTP noise description. 

A pertinent question is how one might come by such a description $\cE$ of the system-bath joint map. One cannot directly measure the bath and so standard process tomography methods do not apply. In addition, one also needs to know the system-bath state at some (usually initial) time, to know what happens to the system at a later time upon application of the system-bath $\cE$. Again, one cannot perform \emph{state} tomography to learn that system-bath state. Such questions are not unique to our work but arise in every open-system discussion where the system is explicitly coupled to a bath. Usually, one declares that one knows the physics of the underlying interaction and can hence write down a system-bath Hamiltonian $\HSB$ for the dynamical evolution. One can take a step further and analyze the coupling strengths to arrive at an appropriate split into the small and large baths, and thereby derive from $\HSB$, a joint system-(small-)bath noise map $\cE$. There are some recent developments \cite{Gessner+Breuer:11, Gessner+Breuer:13, Gessner+et.al:14, Modi12, Ringbauer+et.al15, Dai+et.al:16, Paz-Sliva+Viola:14, Norris+et.al:16} that indicate the possibility of obtaining information about the system-bath interaction, initial correlations or initial state from monitoring only the system, at least for the parts of $\cE$ that have an effect on the system. However, much more work needs to be done before we can talk about process-tomography-type characterization of a system-bath $\cE$. In the following, we will simply assume we are given a description of $\cE$, and leave as an open question---important for all work discussing system-bath interactions---how one goes about directly measuring and identifying this $\cE$.

\section{Code preliminaries}\label{sec:Code}
We encode a qudit of information in a system subspace $\cC\subseteq\HS$ of dimension $d ~[\leq \dim(\HS)]$. The full system-bath Hilbert space has the structure
\begin{equation}
\label{eq:H}
\cH=\HS\otimes\HB=(\cC\oplus\cC^\perp)\otimes\HB\equiv \cP\oplus \cP^{\perp},
\end{equation}
where $\cP\equiv \cC\otimes \HB$. We denote the projector onto $\cC$ by $\Pc$; the projector onto $\cP$ is then $P=\Pc\otimes \idB$. The dimension of the bath is written as $d_\mathrm{B}$.
The encoded information is stored in $\cC$ only, i.e., as a system state (density operator) $\rhoS$ supported on $\cC$. We denote the set of states supported on a subspace $\cU$  by $\sS(\cU)$. The set of (system-only) code states is hence $\sS(\cC)$. Viewed as a system-bath state, a code state $\rhoS\in\sS(\cC)$ is any state $\rho\in\sS(\cH)$ such that $\rhoS=\trb\{\rho\}$, the partial-trace over the bath. All such states $\rho$ satisfy $\rho=P\rho P$. The extension to a system-bath state $\rho$ is nonunique, but they all carry the same encoded information $\rhoS$. We refer to the system-only $\cC$ and its extension to the bath $\cC\otimes\HB$ both as the code space, and $\sS(\cC)$ and $\sS(\cC\otimes\HB)$ as the set of code states. Whether we mean the system-only or system-bath version will be clear from the context. Note that the system-bath code states can generally be entangled across the system-bath cut.

We consider the action of $\cE$ on states initialized in the code space, $\rho\in\sS(\cC\otimes\HB)$. The range of $\cE$ for inputs from the code space is denoted by $\Pe$; equivalently, $\Pe$ is the support of $\cE(P)$. $P_\cE$ denotes the projector onto $\Pe$.

Our goal here is to examine how well the encoded information can be protected from noise with the aid of a recovery operation. The recovery operation is a linear CPTP map $\cR$ that takes a state in $\Pe$ back into the code space, i.e., $\cR:\sB(\Pe)\longrightarrow \sB(\cC\otimes\HB)$. It is performed after every application of the noise $\cE$ on a code state, to mitigate its effects. As the bath is inaccessible, the recovery can act nontrivially only on the system, i.e. $\cR$ satisfies a structural constraint,
\begin{equation}\label{StrConstr}
\cR=\cR_\mathrm{S}\otimes \idB.
\end{equation}
All recovery maps considered here are assumed to satisfy this condition. That $\cR$ maps all states back into the code space then amounts to the condition that $\Rs$ satisfies $\Rs=\cP_\cC\circ\Rs$, where $\cP_\cC(\,\cdot\,)\equiv\Pc(\,\cdot\,)\Pc$ is the projection map onto $\cC$.

For an initial code state $\rho$, the state after the noise and followed by the recovery is $(\cR\circ\cE)(\rho)$. We say that the noise $\cE$ is \emph{approximately correctable} on code $\mathcal{C}$ if there is a CPTP recovery $\cR$ such that 
\begin{equation}
\trb\{(\cR\circ\cE)(\rho)\}\simeq \trb\{\rho\}, \quad \forall \rho\in \sS(\cC\otimes\HB).
\end{equation}
The recovery is assumed to map every state back into the code space, so that $\cR(\tau)=P\big(\cR(\tau)\big)P$, at least for all $\tau\in\sS(\Pe)$.
Given code $\cC$ and noise $\cE$, we quantify the quality of correction for a recovery $\cR$ by the fidelity between the initial state $\rho$ and the recovered state $\sigma\equiv(\cR\circ\cE)(\rho)$: $F(\rho,\sigma)\equiv\tr\sqrt{\rho^{1/2}\sigma\rho^{1/2}}$. The fidelity loss---the deviation of the square of the fidelity from 1--- for recovery $\cR$ on state $\rho$, is
\begin{equation}
\eta_\cR(\rho)\equiv 1-F^2\big(\trb\{\rho\}, \trb\{(\cR\circ\cE)(\rho)\}\big),
\end{equation}
and we quantify the performance of error correction by the \emph{worst-case fidelity loss} over all code states, i.e.,
\begin{equation}
\eta_\cR(\cC) \equiv \max_{\rho\in  \sS(\cC\otimes\HB)} \eta_\cR(\rho).
\end{equation}
The concavity of the fidelity measure \cite{Nielsen+Chuang:10} means that maximal $\eta_\cR$ is attained on a pure \emph{system-bath} state, i.e., the maximization above over the code states $\rho\in\sS(\cC\otimes\HB)$ can be restricted to pure states $\rho$ only.

The \emph{optimal recovery}, $\cR_\mathrm{op}$, is the one with the smallest fidelity loss among all possible recovery maps. Its fidelity loss is denoted as $\eta_\mathrm{op}(\cC)$,
\begin{equation}
\eta_\mathrm{op}(\cC) \equiv \eta_{\cR_\mathrm{op}}(\cC) =\min_{\textrm{CPTP }\cR}\eta_\cR(\cC).
\end{equation}
Following Refs.~\cite{HK+Prabha:10, HK+Prabha:12}, we say that $\cE$ is $\epsilon$-correctable on code $\cC$ if $\epsilon\geq \eta_\mathrm{op}(\cC) $.
$\cE$ is perfectly correctable for code $\cC$ if $\epsilon =0$, i.e., $\eta_\mathrm{op}(\cC)=0$.  
The minimization over CPTP $\cR$ above is understood to be over those $\cR$ satisfying the structural constraint \eqref{StrConstr}. 

A notational remark: Let $\{\ket{k}\}_{k=1}^{d_\mathrm{B}}$ be an orthonormal basis for $\HB$. For any system-bath operator $O$, we define 
\begin{equation}
O_{;k\ell}\equiv \bra{k}O\ket{\ell},
\end{equation}
a system-only operator. We  write $O_{;k\ell}^\dagger \equiv (O_{;k\ell})^\dagger = \bra{k}O\ket{\ell}^\dagger = \bra{\ell}O^\dagger\ket{k}=(O^\dagger)_{;\ell k}$. To avoid overloading the notation, when the meaning is clear from the context, we sometimes drop the semi-colon and write simply $O_{k\ell}$.

\section{Perfect OSQEC conditions}\label{sec:PerfQEC}
We begin with the conditions for perfect QEC, for a CPTP noise map $\cE$ on the system and bath. These conditions are the open-system analogue of what is often referred to as the Knill-Laflamme QEC conditions \cite{Knill+Laflamme:97} for a system-only noise map.

\begin{theorem}[OSQEC conditions]\label{PerfQEC}
A CPTP noise $\cE\sim\{E_a\}_{a=1}^N$ on the system and the bath is perfectly correctable on a code $\cC$ with a CPTP recovery $\cR=\Rs\otimes\idB$ if and only if
\begin{equation}\label{perfectQEC}
\Pc E_{a;k\ell}^\dagger E_{b;mn}\Pc=\Lambda_{ak\ell,bmn}\Pc
\end{equation}
for all $a,b=1,2,\ldots, N$, and $k,\ell,m,n=1,2,\ldots,\dB$. Here, $\Lambda$ is a Hermitian matrix for the triple indices $ak\ell$ and $bmn$: $(\Lambda^\dagger)_{ak\ell,bmn}=(\Lambda_{bmn,ak\ell})^*$.
\end{theorem}

\begin{proof}
``$\Rightarrow$". Suppose $\cE$ is perfectly correctable on $\cC$, i.e., there exists a CPTP recovery $\cR=\Rs\otimes\idB$, with $\Rs \sim\{R_c\}_{c=1}^K$,  such that
\begin{equation}\label{eq:thm1}
\trb\{ (\cR\circ\cE)(\rho)\}=\Rs\big(\trb\{\cE(\rho)\}\big)=\trb\{\rho\}
\end{equation}
for any $\rho\in \sS(\cC\otimes\HB)$. Note that $\rho=P\rho P$, and $\Pc R_c=R_c$. 
The left-hand side of Eq.~\eqref{eq:thm1} amounts to the action, on $\rho$, of a quantum channel with Kraus operators $\{R_c \bra{m} E_a P\}$; the right-hand side is the channel with Kraus operators $\{\Pc\bra{n}=\bra{n}P\}$. The two Kraus representations for the same quantum channel must be unitarily related, i.e., $R_c\bra{m}E_bP=\sum_{n=1}^{d_\text{B}}u_{n,mcb}\bra{n}P$ for a unitary $u$. As $\Rs$ is TP, we also have $\sum_{c=1}^K R_c^\dagger R_c=\idS$. Thus, we have
\begin{align}
\Pc E_{a;k\ell}^\dagger E_{b;mn}\Pc
=& \Pc\bra{\ell}E_a^\dagger\ket{k}\Bigl(\sum_c\! R_c^\dagger R_c\!\Bigr)\bra{m}E_b\ket{n}\Pc  \nonumber\\
=&\sum_c u_{\ell,kca}^* u_{n,mcb} \Pc\equiv\Lambda_{ak\ell,bmn}\Pc.
\end{align}

\noindent ``$\Leftarrow$". Suppose Eq.~\eqref{perfectQEC} holds. The Hermitian $\Lambda$ matrix can be diagonalized with a unitary $v$, 
\begin{equation}
\Lambda_{ak\ell,bmn}=\sum_{a'k'\ell'}v_{ak\ell,a'k'\ell'}\lambda_{a'k'\ell'}v_{bmn,a'k'\ell'}^*.
\end{equation}
By defining $F_{ak\ell}\equiv\sum_{bmn}v_{bmn,ak\ell}E_{b;mn}$, or equivalently, $E_{b;mn}\equiv\sum_{ak\ell}v_{bmn,ak\ell}^*F_{ak\ell}$,
the QEC condition can be written in the diagonal form \cite{note2},
\begin{equation}\label{eqFakl2}
\Pc F_{ak\ell}^\dagger F_{bmn}\Pc=\lambda_{ak\ell}\delta_{ak\ell,bmn}\Pc\quad\forall a,b,k,\ell,m,n. 
\end{equation}
Now, let $R_{ak\ell}\equiv\Pc F_{ak\ell}^\dagger/\sqrt{\lambda_{ak\ell}}$, and regard these as the Kraus operators for a CPTP (recovery) map $\Rs$ on the system; let $\cR\equiv \Rs\otimes\idB$. Writing $\rho=\sum_{k,\ell}\rho_{k\ell}\otimes\ket{k}\bra{\ell}$, where $\Pc\rho_{k\ell}\Pc=\rho_{k\ell}$, we have
\begin{align}
&\trb\{ (\mathcal{R}\circ\mathcal{E})(\rho)\}=\Rs\Big(\sum_m \bra{m}\mathcal{E}(\rho)\ket{m} \Big)\\
=&\sum_{bpq}\sum_{ak\ell m} R_{bpq}E_{a;mk} \Pc\rho_{k\ell} \Pc  E_{a;m\ell}^\dagger R_{bpq}^\dagger\nonumber\\
=&\sum_{abk\ell mpq}  v_{am\ell,bpq} \lambda_{bpq}v_{amk,bpq}^* \Pc\rho_{k\ell}\Pc\nonumber\\
=&\sum_{ak\ell m} \Lambda_{am\ell,amk} \Pc \rho_{k\ell}\Pc\nonumber\\
=&\sum_{ak\ell m} \Pc E_{a;m\ell}^\dagger E_{a;mk} \rho_{k\ell}\Pc  \nonumber\\
=& \sum_{k\ell} \Pc\bra{\ell}  {\biggl[\sum_a E_a^\dagger \Big(\sum_m \ket{m}\bra{m}\Big) E_a\biggr]} \ket{k} \rho_{k\ell}\Pc \nonumber \\
=&\sum_{k,\ell=1}^{d_\text{B}} \Pc\rho_{k\ell}\Pc \delta_{k\ell}= \trb\{\rho\},\nonumber
\end{align}
where, in the fifth line, we used Eq.~\eqref{perfectQEC}, and in the sixth, the fact that $\cE$ is TP.
\end{proof}

Note that the OSQEC conditions also hold for a CP noise map $\mathcal{E}\sim\{E_a\}$ that is sub-TP (at least on the code space), i.e., $P\bigl(\sum_a E_a^\dagger E_a\bigr) P=\gamma^2 P$, $\gamma$ a constant so that $\tr\{\cE(\rho)\}=\gamma^2\tr\{\rho\}$ same for any code state $\rho$. In this case, under perfect QEC, the recovery undoes the noise up to the factor $\gamma^2$: $\trb\{ (\mathcal{R}\circ\mathcal{E})(\rho)\}=\gamma^2\trb\{\rho\}$. Below, we will discuss only CPTP maps, in preparation for the generalization to the approximate case, but the results within this section on perfect QEC are all applicable to this sub-TP situation by the replacement $E_a\rightarrow \frac{1}{\gamma} E_a$ in all statements. In particular, condition \eqref{perfectQEC} looks exactly the same, except that $\Lambda$ now contains a factor of $\gamma^2$.

We mention here the relationship between our OSQEC conditions and the conditions for OQEC, or subsystem codes. In OQEC, one considers a system with a Hilbert space structure
\begin{equation}
\cH=\cQ\oplus\cQ^\perp\equiv(\cH_\mathrm{A}\otimes\cH_\mathrm{B})\oplus \cQ^\perp,
\end{equation}
where $\cQ\equiv \cH_\mathrm{A}\otimes\cH_\mathrm{B}$ has a tensor-product structure. One encodes information into one of the two factors (subsystems) of $\cQ$, say $\cH_\mathrm{A}$. $\cQ$ is the code subspace, but the information is carried only by $\cH_\mathrm{A}$, i.e., two code states supported on $\cQ$ encode the same information if they have the same state on subsystem A. $\cQ$ is said to be correctable for a CPTP noise $\mathcal{E}$ on $\mathcal{H}$ if there exists a CPTP map $\mathcal{R}$ on $\mathcal{H}$ such that 
\begin{align}
\trb\{(\cR\circ\cE)(\rho)\}=\rho_\mathrm{A},
\end{align}
for $\rho\equiv \rho_\mathrm{A}\otimes\rho_\mathrm{B}$, where $\rho_\mathrm{A}$ and $\rho_\mathrm{B}$ are states on subsystems A and B, respectively. Algebraic conditions for (perfect) OQEC are known \cite{Nielsen+Poulin:07, Kribs+Spekkens:06, Choi+et.al.:09}:
\begin{align}
P_\mathrm{A} \bra{\ell} E_i^\dagger E_j\ket{k} P_\mathrm{A} &= \Lambda_{i\ell,jk} P_\mathrm{A}, \quad\forall i,j,k,\ell,\label{eq34}
\end{align}
where $\cE\sim\{E_i\}$, $P_\mathrm{A}$ is the projector onto $\cH_\mathrm{A}$, and $\{\ket{\ell}\}$ is an orthonormal basis for $\HB$.

Structurally, the cases of OSQEC and OQEC are similar: The noise map $\cE$ acts on the full Hilbert space $\cH$, and we store information in a tensor-product factor of a subspace in the full Hilbert space---the $\cC$ factor in the $\cC\otimes\HB$ subspace of $\cH$ for OSQEC, and the $\cH_\mathrm{A}$ factor in the $\cH_\mathrm{A}\otimes\cH_\mathrm{B}$ subspace for OQEC. The crucial difference lies in the fact that in OSQEC, the recovery map $\cR$ can act only on the system, i.e., not on the $\HB$ factor belonging to the uncontrollable or inaccessible bath; in OQEC, $\cR$ can, and will in general, act on both factors, $\cH_\mathrm{A}$ and $\cH_\mathrm{B}$. As such, the OSQEC conditions are more stringent: A code satisfying the OSQEC conditions will satisfy the OQEC conditions, but the converse is not true. Physically speaking, in the setting of OSQEC, we have no access to the bath, and any information that flows there cannot be recovered. In OQEC, however, although the information is carried by subsystem A only, we do have access to and full control on subsystems A and B, and information that flows out to B from A can be recovered, if the OQEC conditions are satisfied.

\medskip
\begin{corollary}\label{GCor}
A CPTP noise $\cE$ is perfectly correctable on code $\cC$ if and only if the system-only channel
\begin{equation}\label{perfectQEC-GCor}
\cG(\,\cdot\,)\equiv(\trb\circ\cE){\left(\,\cdot\,\otimes \tfrac{1}{\dB}\idB\right)}
\end{equation}
is perfectly correctable on the same code $\cC$. In particular, if the recovery for $\cE$ is $\cR=\Rs\otimes\idB$, then the recovery for $\cG$ is $\Rs$.
\end{corollary}

\begin{proof}
Observe that the Kraus operators for $\mathcal{G}$ are simply $G_{ak\ell}\equiv E_{a;k\ell}/\sqrt{\dB}$. The conditions Eq.~\eqref{perfectQEC} are hence equivalent to the standard system-only Knill-Laflamme QEC conditions \cite{Knill+Laflamme:97} for the system-only channel $\mathcal{G}$.
\end{proof}

\noindent Corollary \ref{GCor} establishes a relation between the correctability of our system-bath joint CPTP noise $\cE$, and the correctability of a system-only quantum channel $\cG$. We emphasize, however, that such a relation exists only in quantifying the performance of QEC. $\cG$ is \emph{not} an alternate or effective description of the action of $\cE$ on the system:
\begin{equation}
\cG\bigl(\trb\{\rho\}\bigr)\neq \trb\{\cE(\rho)\}
\end{equation}
in general.

Let us define scaled versions of the operators $F_{ak\ell}$ above: Let $F_\alpha\equiv F_{ak\ell}/\sqrt{\lambda_{ak\ell}}$, writing $\alpha$ in place of the triple-index $ak\ell$. The set $\{F_\alpha\}$ satisfies the orthonormality condition
\begin{equation}\label{eq:ONCond}
\Pc F_\alpha^\dagger F_\beta \Pc=\delta_{\alpha\beta}\Pc\qquad\forall\alpha,\beta.
\end{equation}
The set $\{F_\alpha\Pc\}$ captures the correctability properties of the code $\cC$. The recovery map constructed from this set as $\cR_\mathrm{S}\sim\{\Pc F_\alpha^\dagger\}$, which we had used already in our proof of Theorem \ref{PerfQEC}, has universal properties---we hence call it the \emph{universal recovery}---as captured in the following Lemma.

\medskip
\begin{lemma}[OSQEC conditions on linear span]\label{SpanCor}
Suppose a CPTP noise $\cE\sim \{E_a\}$ is perfectly correctable on code $\cC$. Then, any CPTP noise $\widetilde{\cE}\sim \{\widetilde{E}_b\}$ such that $\widetilde{E}_b P\in \mathrm{span}\{E_a P\}$ is also perfectly correctable on $\cC$. Moreover, both are correctable using the same recovery $\cR=\cR_\mathrm{S}\otimes \idB$, with $\cR_\mathrm{S}\sim\{\Pc F_\alpha^\dagger\}$, with the $F_\alpha$s as defined above for $\mathcal{E}$. 
\end{lemma}

\begin{proof}
Define the system-only channels $\cG(\,\cdot\,)\equiv(\trb\circ\cE)\bigl(\,\cdot\,\otimes \tfrac{1}{\dB}\idB\bigr)$ and 
$\widetilde{\cG}(\,\cdot\,)\equiv(\trb\circ\widetilde{\cE})\bigl(\,\cdot\,\otimes \tfrac{1}{\dB}\idB\bigr)$, with their respective sets of Kraus operators $\cG\sim\{E_{a;k\ell}/\sqrt{\dB}\}$ and $\widetilde{\cG}\sim\{\widetilde{E}_{b;mn}/\sqrt{\dB}\}$. Since $\widetilde{E}_b P\in \mathrm{span}\{E_a P\}$, it follows that $\widetilde{E}_{b;mn}\Pc\in \mathrm{span}\{E_{a;k\ell} \Pc\}=\mathrm{span}\{F_\alpha \Pc\}$. Now, if a system-only quantum channel  $\mathcal{E}_\mathrm{S}\sim\{S_a\}$ is perfectly correctable on $\cC$, any system-only quantum channel $\widetilde{\mathcal{E}}_\mathrm{S}\sim\{\widetilde{S}_a\}$ where $\widetilde{S}_a\in\mathrm{span}\{S_a\}$ is also correctable on $\cC$ by the same recovery operation \cite{Nielsen+Chuang:10}. As $\cG$ is correctable by Corollary \ref{GCor}, it follows that $\widetilde{\cG}$ is correctable. By Corollary \ref{GCor} once more, $\widetilde{\cE}$ thus is correctable, with $\cR_\mathrm{S}\sim\{\Pc F_\alpha^\dagger\}$.
\end{proof}

Turning the logic of the proof of Lemma \ref{SpanCor} around, one can alternatively say that, given any set of $F_\alpha$ operators orthonormal on $\Pc$, i.e., satisfies condition \eqref{eq:ONCond}, any CP map $\cE\sim\{E_a\}$ such that $E_{a;k\ell}\Pc\in\mathrm{span}\{F_\alpha\Pc\}$ is perfectly correctable on $\cC$, and the recovery map is the universal recovery $\cR_\mathrm{S}\sim\{\Pc F_\alpha^\dagger\}$. Note that the statement $\bra{k}E_a\ket{\ell}\Pc=E_{a;k\ell}\Pc\in\mathrm{span}\{F_\alpha\Pc\}$ holds for any basis $\{|k\rangle\}$ if it holds for one particular basis choice. Hence, the perfect correctablility of $\cE\sim\{E_a\}$ is dependent on its $E_a$s, but not on the choice of the bath basis. One should also observe that Eq.~\eqref{perfectQEC} in Theorem \ref{PerfQEC} is form-invariant under a bath-basis change, and Eq.~\eqref{perfectQEC-GCor} in Corollary \ref{GCor} makes no reference to any bath basis at all. Indeed, it would have been meaningless to talk about the correctability of a system-bath CP noise, if there is such a basis-choice dependence on the inaccessible bath.

\section{Approximate OSQEC}\label{sec:AQEC}
While perfect correctability is certainly desirable, nature unfortunately does not often provide us with such an option. The prototypical example is that of independent noise acting on $n$ physical qubits, where one may find codes so that the perfect QEC conditions are fulfilled, provided one considers errors on no more than some $t<n$ qubits. The real noise though includes terms with errors on more than $t$ qubits, and is at best only approximately correctable on the chosen code designed for $<t$ errors. In this section then, we discuss the more generic situation of approximate OSQEC, where perfect removal of noise is not possible and the optimal fidelity loss $\eta_\textrm{op}(\cC)$ is nonzero. There has been a lot of past work on AQEC, such as the construction of approximate recovery operations, the required conditions, etc., all in the context of system-only noise (see, for example, Refs. \cite{HK+Prabha:10, HK+Prabha:12, Leung+et.al:97, Schumacher+et.al:02, Yamamoto+et.al:05, Reimpell+Werner:05, Kosut+et.al:08, Fletcher+et.al:08:1, Fletcher+et.al:08:2, Tyson:10, Beny+Oreshkov:10, Beny:11, Cafaro+Loock:14:1, Cafaro+Loock:14:2, Flammia+et.al:17, Lihm+et.al:17}). Here, we consider a joint system-bath noise that contains a dominant part that is perfectly correctable on the chosen code. The remaining ``uncorrectable part" is small in comparison. The following theorem quantifies the performance of the code in the approximate OSQEC situation.

\begin{theorem}[Approximate OSQEC]\label{thm:AOSQEC}
Consider a (system-bath) CPTP noise $\cE\sim\{E_a\}$ and a code space $\cC$. Suppose $E_a=\widetilde{E}_a+B_a$ such that $\cC$ is perfectly correctable under the CP $\widetilde{\cE}\sim\{\Et_a\}$. Then, $\cE$ is $\epsilon$-correctable on $\cC$, with 
\begin{equation}\label{eq:AOSQEC}
\epsilon\,\equiv\,\tfrac{1}{8}{\left(1+\sqrt 2\right)} {\left\Vert \Theta^\dagger\Theta\right\Vert} + {\left\Vert \Delta \right\Vert}+O{\left((BP)^3\right)},
\end{equation}
where $\Theta\equiv \sum_a P(B_a^\dagger \Et_a-\Et_a^\dagger B_a)P$, $\Delta\equiv \sum_a  P B_a^\dagger B_a P$,  and $\Vert\,\cdot\,\Vert$ is the operator norm. The $O{\left((BP)^3\right)}$ term indicates that our bound is accurate to second order in the size of the $B_a$ operators, restricted to their action on $P (\equiv\Pc\otimes\idB)$.
\end{theorem}

\begin{proof}
We give here a sketch of the proof, leaving some of the lengthy calculations to the Appendix. $\widetilde{\cE}$ satisfies the perfect QEC conditions on $\cC$. From Sec.~\ref{sec:PerfQEC} then, we have an equivalent set of Kraus operators $\{F_\alpha\Pc\}$  such that $\Pc F_\alpha^\dagger F_\beta\Pc=\delta_{\alpha\beta}\Pc$, and $\bra{k}\Et_a\ket{\ell}\Pc=\Et_{a ;k\ell}\Pc=\sum_\alpha e_{ak\ell,\alpha}F_\alpha \Pc$ for every $a,k,\ell$. The recovery for $\widetilde{\cE}$ on $\cC$ is $\cR\sim\{\Pc F_\alpha^\dagger\}$, which we will also use as the recovery map for $\cE$ on $\cC$. 
In addition, with $\bra{k}B_a\ket{\ell}\Pc=B_{a;k\ell}\Pc$, let $D_{\alpha a;k\ell}\equiv \Pc F_\alpha^\dagger B_{a;k\ell}\Pc$, and write the input code state $\rho=\sum_{k\ell=1}^{d_\text{B}} \rho_{k\ell} \otimes \ket{k}\bra{\ell}$ with $\rho_{k\ell}=\Pc\rho_{k\ell}\Pc$, a system-only operator. Note that $(\rho_{k\ell})^\dagger =\rho_{\ell k}$, by Hermiticity of $\rho$.
Observe that $\Pc F_\alpha^\dagger \Et_{a;k\ell}\Pc=e_{ak\ell,\alpha}\Pc$, and that $\sum_\alpha e_{ak\ell,\alpha}D_{\alpha a;km}^\dagger=\Pc B_{a;km}^\dagger \Et_{a;k\ell}\Pc$.

Now we calculate $\rhoS'\equiv \trb\{(\mathcal{R}\circ\mathcal{E})(\rho)\}$ (see Appendix for details):
\begin{align}
\label{eq23}\rhoS'&=\trb\{(\mathcal{R}\circ\mathcal{E})(\rho)\}=\rhoS+\sigma_1+\sigma_2 
\end{align}
\begin{align}
\textrm{with} \quad \rhoS&\equiv\trb\{\rho\},\nonumber\\
\sigma_1&\equiv\frac{1}{2}\sum_{\ell m}[\rho_{\ell m},\Theta_{m\ell}],\nonumber\\
\textrm{and}\quad\sigma_2&\equiv-\frac{1}{2}\sum_{\ell m}\{\rho_{\ell m},\Delta_{m\ell}\}+\sum_{ak\ell m}\cD_{a;k\ell m}(\rho_{\ell m}).\nonumber
\end{align}
Here, $\cD_{a;k\ell m}$ is the linear map $\cD_{a;k\ell m}(\cdot)\equiv \sum_\alpha D_{\alpha a;k\ell}(\cdot)D_{\alpha a;km}^\dagger$, and $[\,,\,]$ and $\{\,,\,\}$ are the commutator and anti-commutator, respectively. It is easy to check that $\sigma_1$ and $\sigma_2$ are Hermitian operators.

Above, we have grouped the terms in ``powers of $BP$", referring here to the $B_aP$ or $B_{a;k\ell}\Pc$ operators: $\sigma_1\sim\Theta\sim BP$, and $\sigma_2\sim\Delta \sim (BP)^2$. The $\rhoS$ term above is then the 0th-order-in-$BP$ term. We expect our results to be useful for situations where the $BP$ operators are small.

We want an expression for the fidelity between the input code state $\rhoS$ and the output state $\rhoS'$ after noise and recovery: $F\big(\rhoS, \rhoS'\big)=\tr{\left[\sqrt{\rhoS}\rhoS'\sqrt{\rhoS}\right]}^{1/2}$. Our quantity of interest, the minimum fidelity loss $\eta_\mathcal{R}\{\mathcal{C}\}$ requires a minimization over pure states $\rho\equiv \Psi\equiv \ket{\Psi}\bra{\Psi}$ with $\ket{\Psi}\in\cC\otimes\HB$.
Let $M\equiv\sqrt{\rhoS}\rhoS'\sqrt{\rhoS}$ so that $F(\rhoS,\rhoS')=\tr\{M^{1/2}\}$. We calculate $M$ in powers of $BP$; as we will see below, we need an expression accurate to second-order in $BP$. We hence write,
\begin{align}
M^{1/2}&\equiv \rhoS+\delta_1+\delta_2, \label{eq25}
\end{align}
where $\delta_1$ and $\delta_2$ are Hermitian operators that are linear and quadratic in $BP$, respectively. Comparing $(M^{1/2})^2=(\rhoS+\delta_1+\delta_2)^2$ with $M=\sqrt{\rhoS}\rhoS'\sqrt{\rhoS}=\sqrt{\rhoS}(\rhoS+\sigma_1+\sigma_2)\sqrt{\rhoS}$, we match terms of the same order in $BP$, neglecting terms of order $(BP)^3$ and higher:
\begin{subequations}
\label{eq:delta}
\begin{align}
\rhoS\delta_1+\delta_1\rhoS&=\sqrt{\rhoS}\sigma_1\sqrt{\rhoS}\label{eq:delta1}\\
\delta_1^2+\rhoS\delta_2+\delta_2\rhoS&=\sqrt{\rhoS}\sigma_2\sqrt{\rhoS}.\label{eq:delta2}
\end{align}
\end{subequations}
Explicit expressions for $\delta_1$ and $\delta_2$ can be obtained (see Appendix). It turns out that $\delta_1$ is traceless, and 
\begin{align}\label{eq:trdelta2}
0\geq \tr\{\delta_2\}
&\geq -{\left[\tfrac{1}{16}\bigl(1+\sqrt 2\bigr)\Vert\Theta^\dagger\Theta\Vert +\tfrac{1}{2}\Vert\Delta\Vert\right]},
\end{align}
for any pure input state. With $F(\rhoS,\rhoS')=\tr\{M^{1/2}\}=1+\tr\{\delta_1\}+\tr\{\delta_2\}+O\bigl((BP)^3\bigr)$, it follows that 
\begin{align}
\eta_\mathrm{op}(\cC)&\equiv\min_\cR\max_\rho{\left[1-F(\trb\{\rho\},\trb\{(\cR\circ\cE)\}(\rho))^2\right]}\nonumber\\
&\leq\max_\rho{\left[1-F(\rhoS,\rhoS')^2\right]}\nonumber\\
&=\max_\rho{\left[-2(\tr\{\delta_1\}+\tr\{\delta_2\})+O{\left((BP)^3\right)}\right]}\nonumber\\
&\leq \tfrac{1}{8}{\bigl(1+\sqrt 2\bigr)}\Vert\Theta^\dagger\Theta\Vert+\Vert\Delta\Vert\equiv \epsilon,
\end{align}
as given in Eq.~\eqref{eq:AOSQEC}.
\end{proof}

Our approach here to bounds for approximate OSQEC is very different from that used in Refs.~\cite{HK+Prabha:10, HK+Prabha:12} on system-only AQEC. There, AQEC conditions were derived, and the bounds on fidelity loss obtained, by perturbing the perfect QEC conditions, as we do above. However, this was done through the use of a recovery map termed the ``transpose channel" (also often referred to as the Petz recovery). In particular, it was shown in those papers that the transpose channel recovery is identical to the standard recovery map for perfect correction in the perfect QEC case, and that it continues to be a provably near-optimal recovery in the approximate case. This transpose channel is a map constructed from the given (system-only) noise channel and the (system-only) code: For a system-only noise map $\cE_\mathrm{S}\sim\{S_a\}$, its transpose channel for a given code $\cC$ is the map with the Kraus operators $\cR\sim\{\Pc S_a^\dagger\cE_\mathrm{S}(\Pc)^{-1/2}\}$ (with the inverse taken on the support). 

In the OSQEC case, the transpose channel plays less of a role. To start, one cannot use as recovery, the transpose channel corresponding to the full system-bath noise $\cE$. That would generically be a system-bath map, which violates the structural constraint \eqref{StrConstr}. One may imagine using a transpose channel recovery for some effective system-only map. However, as mentioned earlier, there is no single effective system-only map that would capture the action of $\cE$ on the code states, and the action of the noise on the information carried by the system depends on the state of the bath. One might have guessed, given the perfect OSQEC conditions, that a plausible candidate for defining the transpose channel is the $\cG$ map of Corollary \ref{GCor}. However, the transpose channel $\cR_\cG$ for that $\cG$ map would not reverse the effects of the action of $\cE$ on the system for input system-bath states unequal to $\rhoS\otimes\idB/\dB$. In fact, for $\cG'\neq\cG$, $\cR_\cG\circ\cG'$ is generally not even trace-preserving on the code space, as is needed for recovery, even if $\cG'$ is itself trace preserving. Without the ability to define a useful transpose channel recovery, we instead make use of the universal recovery from perfect OSQEC, and directly quantified the code performance in the approximate case. Note that, in the system-only AQEC situation, the universal recovery and the transpose channel can be shown to give the same recovery fidelity to first order in the deviation from the perfect situation.

\section{Conclusion} \label{sec:Conc}
In this work, we examined the performance of QEC in an open-system setting. We described the noise acting on the system as originating from a joint system-bath CPTP map, a framework that nicely includes and interpolates between standard noise models used for QIP. The noise model can manifest non-Markovian effects of conceptual and practical interest for QIP. Specifically, we derived conditions on the noise for the existence of codes that allow for perfect preservation of information with the use of a recovery operation that acts only on the system. This extends past work on QEC conditions for the standard system-only noise channels. Further connections with the system-only QEC situation were made from the observation that one has perfect correctability in the OSQEC case if and only if a system-only noise channel (the $\cG$ map of Corollary \ref{GCor}) is itself perfectly correctable and hence satisfies the standard Knill-Laflamme QEC conditions.

We also derived a bound on the worst-case fidelity for protection by a code under the action of the noise model in the case where only approximate, not perfect, correction is possible. This extends and completes the program for AQEC set out in Refs.~\cite{HK+Prabha:10, HK+Prabha:12}. There is a difference in approach here, however. Rather than the transpose channel recovery used in previous work, approximate recovery is carried out here using the universal recovery built only from the perfectly correctable part of the full noise map. The transpose channel turned out to be difficult to define for the case of OSQEC.

As mentioned in Sec.~\ref{sec:NoiseModel}, our work leaves open the question of identifying and characterizing the system-bath noise map $\cE$. Our derived conditions can only be checked if one possesses the full description of $\cE$. Currently, we can only assume that one has a good understanding of the system and bath dynamics, and can derive a suitable $\cE$ from an underlying microscopic system--(full-)bath Hamiltonian. It would be much more natural, and better aligned with the notion of \emph{experimental} system characterization common in other QIP tasks, if one could figure out how to apply tomography methods to find out about $\cE$. Such a capability will  no doubt have utility not just for our OSQEC problem here, but for all studies of open-system dynamics.

\section*{Acknowledgment}
The authors thank Yi-Cong Zheng, Jing Hao Chai, and Prabha Mandayam for insightful discussions. HKN is supported by Yale-NUS College [through internal grant (MOE Tier-1) IG14-LR001, and a start-up grant]. This work is also supported by the National Research Foundation, and the Ministry of
Education, Singapore.

\section*{APPENDIX}\label{App}
In this Appendix, we provide additional details for the proof of Theorem \ref{thm:AOSQEC}.\\

\noindent \textbf{Derivation of Eq.~\eqref{eq23}.}\\[1ex]
We begin with $\rhoS'\equiv \trb\{(\cR\circ\cE)(\rho)\}$,
\begin{align}\label{eq:rhoSp}
&\rhoS'\equiv \trb\{(\cR\circ\cE)(\rho)\}\\
=&\sum_\alpha \Pc F_\alpha^\dagger \Big( \sum_{ak\ell m} \bra{k} E_a \ket{\ell}\rho_{\ell m}\bra{m}E_a^\dagger\ket{k} \Big) F_\alpha\Pc\nonumber\\
=&\!\!\!\sum_{\alpha a k \ell m}  \!\!\!\Pc F_\alpha^\dagger (\Et_{a;k\ell}\!+\!B_{a;k\ell})\Pc \rho_{\ell m} \Pc(\Et_{a;km}^\dagger\!+\!B_{a;km}^\dagger) F_\alpha \Pc\nonumber\\
=&\sum_{a k \ell m} {\left[\sum_\alpha e_{ak\ell,\alpha} e_{akm,\alpha}^*\rho_{\ell m}+\Pc \Et_{a;km}^\dagger B_{a;k\ell}\Pc\rho_{\ell m}\right.}\nonumber\\
&{\left.\hphantom{\sum_{\alpha a k \ell m} } +\rho_{\ell m} \Pc B_{a;km}^\dagger \Et_{a;k\ell}\Pc+\mathcal{D}_{a;k\ell m}(\rho_{\ell m})\right]}.\nonumber
\end{align}
Now, we make use of the fact that $\cE$ is TP, i.e., $\sum_a E_a^\dagger E_a=\id$, which gives, using $E_a=\Et_a+B_a$,
\begin{align}
\label{eq:TP}
&\sum_{\alpha ak} e_{akm,\alpha}^*e_{ak\ell,\alpha}\Pc=\sum_{ak} \Pc \Et_{a;km}^\dagger \Et_{a;k\ell}\Pc\\
=&\delta_{m\ell}\Pc-\!\sum_{ak}\Pc {\left(\Et_{a;km}^\dagger B_{a;k\ell}+B_{a;km}^\dagger \Et_{a;k\ell}\right)}\Pc-\Delta_{m\ell}.\nonumber
\end{align}
To make use of this expression, we write the first term in the equality of Eq.~\eqref{eq:rhoSp} as a sum of two terms: $\sum_\alpha e_{ak\ell,\alpha} e_{akm,\alpha}^*\rho_{\ell m}=\frac{1}{2}\big(\sum_\alpha e_{ak\ell,\alpha} e_{akm,\alpha}^*\big)\rho_{\ell m}+\frac{1}{2}\rho_{\ell m}\big(\sum_\alpha e_{ak\ell,\alpha} e_{akm,\alpha}^*\big)$. Inserting Eq.~\eqref{eq:TP} into Eq.~\eqref{eq:rhoSp} and simplifying, one obtains the expressions for $\rhoS'=\rhoS+\sigma_1+\sigma_2$ in the main text.

\bigskip
\noindent \textbf{Solving Eq.~\eqref{eq:delta} for $\delta_1$ and $\delta_2$.}\\[1ex]
We write the pure system-bath input code state $\rho\equiv\Psi=|\Psi\rangle\langle\Psi|$ using Schmidt decomposition: $\ket{\Psi}=\sum_k \mu_k\ket{\psi_k}\ket{k}$, where the first tensor factor $\ket{\psi_k}\in\cC$ is a system code state, and $\ket{k}\in\HB$ is a bath state. The $k$ index goes from 1 to $\min\{d,\dB\}$, and we choose the phases of $\ket{k}$s (say) such that $\mu_k\geq 0~\forall k$. The set $\{\ket{\psi_k}\}_{k=1}^{d}$ forms an orthonormal basis for $\cC$, where, if $d>\dB$, $\ket{\psi_k}$s for $k=\dB+1,\ldots,d$ are chosen to complete the basis given the Schmidt set $\ket{\psi_k}_{k=1}^{\dB}$. Similarly, $\{\ket{k}\}_{k=1}^{\dB}$ is an orthonormal basis for $\HB$, again with additional vectors chosen to complete the basis if $\dB>d$. We also introduce the notation, $\ket{\Psi_{k\ell}}\equiv\ket{\psi_k}\ket{\ell}$ (a system-bath state) and $Q_{k\ell}\equiv|\psi_k\rangle\langle\psi_\ell|$. Observe that $Q_{k\ell}Q_{mn}=\delta_{\ell,m}Q_{kn}$, $\rho_{k\ell}\equiv \bra{k}\rho\ket{\ell}=\mu_k\mu_\ell Q_{k\ell}$, $\Pc Q_{k\ell}=Q_{k\ell}\Pc=Q_{k\ell}$, $\sum_{kk}Q_{kk}=\Pc$, $\rhoS=\trb\{\Psi\}=\sum_k\mu_k^2Q_{kk}$, and $\sqrt{\rhoS}=\sum_k\mu_kQ_{kk}$.

The right-hand sides of Eq.~\eqref{eq:delta} can be written as
\begin{align}
\sqrt{\rhoS}\sigma_1\sqrt{\rhoS}&=\frac{1}{2} \sum_{k\ell m} \mu_k\mu_\ell^2\mu_m\Big(Q_{\ell m}\Theta_{m\ell}Q_{kk}+\mathrm{h.c.}\Big)\\
\sqrt{\rhoS}\sigma_2\sqrt{\rhoS}&=\sqrt{\rhoS}\sum_{ak\ell m} \mathcal{D}_{a;k\ell m}(\rho_{lm})\sqrt{\rhoS}\nonumber\\
&\hphantom{\equiv\,}-\frac{1}{2} \sum_{k\ell m}  \mu_k\mu_\ell\mu_m^2 \bigr(Q_{kk}\Delta_{m\ell}Q_{\ell m}+\mathrm{h.c.}\bigr).\nonumber
\end{align}
One can check that the following expression for $\delta_1$ solves Eq.~\eqref{eq:delta1} (the order-$BP$ equation):
\begin{equation}\label{eq:d1}
\delta_1=\frac{1}{2}\sum_{k\ell}  \frac{\mu_k^2\mu_\ell}{s_{k\ell}}
\Bigl(Q_{k\ell}\bra{\Psi}\Theta\ket{\Psi_{\ell k}}+\mathrm{h.c.}\Bigr),
\end{equation}
where $s_{k\ell}\equiv \mu_k^2+\mu_\ell^2=s_{\ell k}$. In the above sum (and in the ones below), a ($k$,$\ell$) term is regarded as zero whenever $s_{k\ell}=0$. As $\Theta^\dagger =-\Theta$, it is easy to see that $\tr(\delta_1)=0$. The $\delta_1$ term hence gives vanishing contribution to the fidelity expression, and we need the next-order term $\delta_2$.

To solve Eq.~\eqref{eq:delta2} given our solution for $\delta_1$ above, we split $\delta_2$ into two terms, $\delta_2\equiv V+W$ such that
\begin{align}
\rhoS V+V\rhoS&=-\delta_1^2\nonumber\\
\textrm{and}\quad \rhoS W+W\rhoS&=\sqrt{\rhoS}\sigma_2\sqrt{\rhoS},\label{eq:V}
\end{align}
and hence satisfy Eq.~\eqref{eq:delta2}. Straightforward calculations yield the solution $V=\frac{1}{4}(V_1+V_2+V_3)$ and $W=W_1+W_2$ where
\begin{align}
V_1&\equiv \sum_{k\ell m}\frac{-\mu_k^3\mu_{\ell}^2\mu_m}{s_{k\ell} s_{\ell m}s_{mk}}\Big(\bra{\Psi}\Theta\ket{\Psi_{k\ell}}\bra{\Psi}\Theta\ket{\Psi_{m k}}Q_{\ell m}+\text{h.c}\Big),\nonumber\\
V_2&\equiv \sum_{k\ell m}\frac{\mu_k^2\mu_\ell^2\mu_m^2}{s_{k\ell} s_{\ell m}s_{mk}} \bra{\Psi}\Theta\ket{\Psi_{k\ell}}\bra{\Psi_{km}}\Theta\ket{\Psi}Q_{\ell m},\nonumber\\
V_3&\equiv \sum_{k\ell m}\frac{\mu_k\mu_\ell\mu_m^4}{s_{k\ell} s_{\ell m}s_{mk}}\bra{\Psi_{km}}\Theta\ket{\Psi}\bra{\Psi}\Theta\ket{\Psi_{\ell m}}Q_{k\ell},
\end{align}
and 
\begin{align}\label{eq:W}
W_1&\equiv\frac{1}{2}\sum_{k\ell}\frac{-\mu_k\mu_\ell^2}{s_{k\ell}}\Big(\bra{\Psi_{k\ell}}\Delta\ket{\Psi}Q_{k\ell}+\mathrm{h.c.}\Big),\\
W_2&\equiv\sum_{\alpha ak\ell m} \frac{\mu_\ell\mu_m}{s_{\ell m}}Q_{\ell\ell} F_\alpha^\dagger \bra{k}B_a|\Psi\rangle\langle\Psi| B_a^\dagger\ket{k}F_\alpha Q_{mm}.\nonumber
\end{align}
The sums over the indices of the $\mu$ coefficients are understood to be only over those $k$ values for which $\mu_k\neq 0$.

To compute the fidelity, we need the trace of $\delta_2$. We do this for each $V$ and each $W$. First the $W$'s:
\begin{align}
\tr\{W_1\}&= -\frac{1}{2}\langle\Psi|\Delta|\Psi\rangle=-\frac{1}{2}\langle\Psi|{\Bigl(\sum_aB_a^\dagger B_a\Bigr)}|\Psi\rangle,\\
\tr\{W_2\}&=\frac{1}{2}\langle\Psi|{\Bigl[\sum_aB_a^\dagger {\left(\Pt\otimes\idB\right)}B_a\Bigr]}|\Psi\rangle,
\end{align}
where $\Pt\equiv \sum_{\alpha \ell}F_\alpha Q_{\ell\ell}F_\alpha^\dagger$, with the sum over $\ell$ going over those $\ell$s such that $\mu_\ell\neq 0$. The trace of $W=W_1+W_2$ is then
\begin{align}
\tr\{W\}&=-\frac{1}{2}\langle\Psi|{\Bigl[\sum_aB_a^\dagger {\left(\id-\Pt\otimes\idB\right)}B_a\Bigr]}|\Psi\rangle,\nonumber\\
&\geq-\frac{1}{2}\langle\Psi|{\Bigl(\sum_aB_a^\dagger B_a\Bigr)}|\Psi\rangle
\geq-\frac{1}{2} \Vert\Delta\Vert.  \label{eq69}
\end{align}

For the trace of $V$, we have
\begin{align}
\tr\{V_1\}&\equiv -\frac{1}{2}\sum_{k\ell}\frac{\mu_k^3\mu_{\ell}}{s_{k\ell}^2}\Big(\avg{k\ell}\avg{\ell k}+\mathrm{c.c}\Big),\nonumber\\
\tr\{V_2\}&\equiv -\frac{1}{2}\sum_{k\ell}\frac{\mu_k^2\mu_\ell^2}{s_{k\ell}^2} |\avg{k\ell}|^2,\nonumber\\
\tr\{V_3\}&\equiv -\frac{1}{2}\sum_{k \ell}\frac{\mu_\ell^4}{s_{k \ell}^2}|\avg{k\ell}|^2,
\end{align}
where we have used the shorthand $\avg{k\ell}\equiv \bra{\Psi}\Theta\ket{\Psi_{k\ell}}$.
The traces for $V_2$ and $V_3$ are manifestly negative, but $\tr\{V_1\}$ could be positive. The total trace for $V$, however, is in fact non-positive. To see this, observe that 
\begin{equation}
|\tr\{V_2+V_3\}|
=\frac{1}{4}\sum_{k\ell}\frac{1}{s_{k\ell}}{\left(\mu_\ell^2|\avg{k\ell}|^2+\mu_k^2|\avg{\ell k}|^2\right)},
\end{equation}
while, by triangle inequality, we have 
\begin{align}
|\tr\{V_1\}|&\leq \sum_{k\ell}\frac{\mu_k^3\mu_{\ell}}{s_{k\ell}^2}|\avg{k\ell}|~|\avg{\ell k}|\\
&=\frac{1}{2}\sum_{k\ell}\!\!{\left(\!\frac{\mu_k^3\mu_{\ell}}{s_{k\ell}^2}+\frac{\mu_\ell^3\mu_{k}}{s_{k\ell}^2}\!\right)}|\avg{k\ell}|~|\avg{\ell k}|\nonumber\\
&=\frac{1}{2}\sum_{k\ell}\frac{\mu_k\mu_\ell}{s_{k\ell}}|\avg{k\ell}|~|\avg{\ell k}|.\nonumber
\end{align}
We hence see that
\begin{align}
&|\tr\{V_2+V_3\}|-|\tr\{V_1\}|\nonumber\\
\geq&\frac{1}{4}\sum_{k\ell}\frac{1}{s_{k\ell}}\bigl(\mu_\ell|\avg{k\ell}|-\mu_k|\avg{\ell k}|\bigr)^2\geq0,
\end{align}
and thus conclude that $\tr\{V\}=\frac{1}{4}(\tr\{V_1\}+\tr\{V_2+V_3\})\leq 0$. In addition, we can bound,
\begin{align}
{\left\vert \tr\{V\}\right\vert} &\leq \frac{1}{4}{\left(|\tr\{V_2+V_3\}|+|\tr\{V_1\}|\right)}  \nonumber\\
&\leq  \frac{1}{16} \sum_{k\ell} \frac{1}{s_{k\ell}}\bigl(\mu_\ell|\avg{k\ell}|+\mu_k|\avg{\ell k}|\bigr)^2.
\end{align}

A more useful bound is one that does not depend on the Schmidt coefficients $\mu_k$s or on the choice of bath basis.
For that, first note that $|\avg{k\ell}|\,|\avg{\ell k}|\leq \frac{1}{2}(|\avg{k\ell}|^2+|\avg{\ell k}|^2)$. Then, $\bigl(\mu_\ell|\avg{k\ell}|+\mu_k|\avg{\ell k}|\bigr)^2\leq \mu_\ell^2|\avg{k\ell}|^2+\mu_k^2|\avg{\ell k}|^2+\mu_k\mu_\ell(|\avg{k\ell}|^2+|\avg{\ell k}|^2)$. Writing $x_{k\ell}\equiv \mu_k/\mu_\ell (>0)$, and switching the indices $k\leftrightarrow \ell$ where necessary, we arrive at the expression,
\begin{align}
{\left\vert \tr\{V\}\right\vert}
&\leq  \frac{1}{8} \sum_{k\ell} \frac{1}{s_{k\ell}}\bigl(\mu_\ell^2+\mu_k\mu_\ell)|\avg{k\ell}|^2\nonumber\\
&= \frac{1}{8} \sum_{k\ell} \frac{1+x_{k\ell}}{1+x_{k\ell}^2}|\avg{k\ell}|^2.
\end{align}
The maximum value of $\frac{1+x_{k\ell}}{1+x_{k\ell}^2}$ is $\frac{1}{2}(1+\sqrt 2)$, attained when $x_{k\ell}=\sqrt{2}-1$. Together with the fact that $\sum_{k\ell} \ket{\Psi_{k\ell}}\bra{\Psi_{k\ell}}$ for any code state is a  projector onto a subspace of $\cC\otimes\HB$, we have
\begin{align}
\Big|\tr\{V\}\Big| &\leq \frac{1+\sqrt{2}}{16} \bra{\Psi}\Theta^\dagger \Theta\ket{\Psi}\leq \frac{1+\sqrt{2}}{16} \Vert\Theta^\dagger \Theta\Vert.
\end{align}
Recalling that $\tr\{V\}\leq 0$, we thus have $0\geq \tr\{V\}\geq -\frac{1+\sqrt 2}{16}\Vert\Theta^\dagger\Theta\Vert$. Together with the trace of $W$, we finally have
\begin{align}
0\geq \tr\{\delta_2\}&=\tr\{V\}+\tr\{W\}\nonumber\\
&\geq -{\left[\tfrac{1}{16}(1+\sqrt 2)\Vert\Theta^\dagger\Theta\Vert +\tfrac{1}{2}\Vert\Delta\Vert\right]},
\end{align}
as given in Eq.~\eqref{eq:trdelta2}.\\

\textbf{Positivity of $M^{1/2}$ and alternate expressions for $\delta_1$ and $\delta_2$}

To compute $M^{1/2}$, we introduced in Eq.(\ref{eq25}) the \emph{ansatz} $M^{1/2}=\rhoS+\delta_1+\delta_2$, and provided expressions for $\delta_1$ and $\delta_2$ accurate to quadratic order in $BP$. We did not, however, explicitly ensure that our expressions for $\delta_1$ and $\delta_2$ give a non-negative $M^{1/2}$, as is understood in the definition of the square root of $M$ for the evaluation of the fidelity. Our $M^{1/2}$ expression \emph{is} manifestly Hermitian, but there is no \emph{a priori} guarantee of nonnegative eigenvalues.

Instead, we show nonnegativity of our expression for $M^{1/2}$, \emph{a posteriori}. In particular, we show here that it is equal, up to quadratic order in $BP$, to an operator writable as an exponential, which guarantees nonnegativity,
\begin{equation}\label{eq74}
M^{1/2}=\rhoS+\delta_1+\delta_2=\upe^{A+\dA+\ddA}, 
\end{equation}
where $A, \dA$ and $\ddA$ are Hermitian operators, zeroth, first, and second order, respectively, in $BP$. 

We first expand the RHS of Eq.(\ref{eq74}) in powers of $BP$ using standard perturbation theory. To second order in $BP$, we have (see, for example, the Appendix of Ref.~\cite{Karplus+Schwinger:48}),
\begin{align}
&\upe^{A+\dA+\ddA}\nonumber\\
\simeq&\,\upe^A + \int_0^1 \mathrm{d}\alpha\, \upe^{(1-\alpha)A}\dA\,\upe^{\alpha A} +\int_0^1 \mathrm{d}\alpha\, \upe^{(1-\alpha)A}\ddA\,\upe^{\alpha A} \nonumber\\
&+ \int_0^1 \mathrm{d}\alpha\int_0^\alpha \mathrm{d}\alpha'\, \upe^{(1-\alpha)A}\dA\,\upe^{(\alpha-\alpha') A}\dA\,\upe^{\alpha' A},
\end{align}
and identify 
\begin{align}
\rhoS=&\,\upe^A,\\
\delta_1=&\,\int_0^1 \mathrm{d}\alpha\, \upe^{(1-\alpha)A}\dA\,\upe^{\alpha A}, \label{eq77} \\
\delta_2=&\,\int_0^1 \mathrm{d}\alpha\, \upe^{(1-\alpha)A}\ddA\,\upe^{\alpha A} \nonumber\\
&+\int_0^1 \mathrm{d}\alpha\int_0^\alpha \mathrm{d}\alpha'\, \upe^{(1-\alpha)A}\dA\,\upe^{(\alpha-\alpha') A}\dA\,\upe^{\alpha' A}. \label{eq78} 
\end{align}
The operators $\dA$ and $\ddA$ are in turn identified from the consistency requirement, up to 2nd order in the uncorrectable part, 
\begin{align}
M=\rhoS^2+\rhoS^{1/2} \sigma_1 \rhoS^{1/2}+\rhoS^{1/2} \sigma_2 \rhoS^{1/2}=\upe^{2(A+\dA+\ddA)}, \label{eq79}
\end{align}
that is, we identify 
\begin{align}
\rhoS^2=&\,\upe^{2A},\quad \textrm{i.e.,}\quad \rhoS=\upe^{A},\\
\rhoS^{1/2} \sigma_1 \rhoS^{1/2}=&\,2\int_0^1 \mathrm{d}\alpha\, \upe^{(1-\alpha)2A}\dA\,\upe^{2\alpha A},  \label{eq81} \\
\rhoS^{1/2} \sigma_2 \rhoS^{1/2}=&\,2\int_0^1 \mathrm{d}\alpha\, \upe^{(1-\alpha)2A}\ddA\,\upe^{2\alpha A}\label{eq82}\\ 
&\hspace*{-1cm}+4\int_0^1 \!\!\mathrm{d}\alpha\!\int_0^\alpha \!\!\!\mathrm{d}\alpha'\, \upe^{(1-\alpha)2A}\dA\,\upe^{(\alpha-\alpha') 2A}\dA\,\upe^{2\alpha' A}. \nonumber 
\end{align}
From before, $\upe^A=\rhoS=\sum_k\mu_k^2Q_{kk}$, with the sum over $k$ taken for those with $\mu_k>0$ only. The left-hand side of Eq.~\eqref{eq81} is then 
\begin{align}
\sum_{k\ell} \mu_k\mu_\ell (\sigma_1)_{k\ell}Q_{k\ell} ,\label{eq83}
\end{align}
where we make use of the shorthand $(\cdot)_{k\ell}\equiv\bra{\psi_k}(\cdot)\ket{\psi_\ell}$, and the right-hand side is
\begin{align}
&2\sum_{k\ell} Q_{k\ell} \int_0^1\upd\alpha\, \mu_k^4 {\left(\frac{\mu_\ell^4}{\mu_k^4}\right)}^{\alpha}(\dA)_{k\ell}\nonumber\\
=&\frac{1}{2}\! \sum_{k\ell} \frac{(x_{\ell k}^4-1)}{\ln (x_{\ell k})}\mu_k^4 Q_{k\ell}  (\dA)_{k\ell}, \label{eq84}
\end{align}
recalling our earlier notation: $x_{k\ell}\equiv \mu_k/\mu_\ell$. For the factor $\frac{(x_{\ell k}^4-1)}{\ln (x_{\ell k})}$ above, the limiting value of 4 is understood when $x_{\ell k}=1$.
From Eqs. \eqref{eq83} and \eqref{eq84}, we identify the matrix elements of $\dA$,
\begin{equation}
(\dA)_{k\ell}=2\frac{\mu_\ell}{\mu_k^3}\frac{\ln(x_{\ell k})}{(x_{\ell k}^4-1)} (\sigma_1)_{k\ell}. \label{eq86}
\end{equation}
Since $\sigma_1$ is Hermitian, so is $\dA$. In addition, the factor between their matrix elements is invariant under the exchange $k\leftrightarrow \ell$. A similar calculation for Eq.~\eqref{eq77} yields 
\begin{align}
(\delta_1)_{k\ell}&=\frac{1}{2}\mu_k^2\frac{x_{\ell k}^2-1}{\ln(x_{\ell k})} (\delta A)_{k\ell}, \label{eq87}
\end{align}
where, again, the limit (of 2 this time) is understood when $x_{\ell k}\rightarrow 1$.
Combining the last two equations gives
\begin{align}
(\delta_1)_{k\ell}&=\frac{\mu_k\mu_\ell}{s_{k\ell}} (\sigma_1)_{k\ell}, \label{eq88}
\end{align}
(recall, $s_{k\ell}=\mu_k^2+\mu_\ell^2$) a relationship satisfied by our earlier expressions of $\delta_1$ [Eq.~\eqref{eq:d1}] and $\sigma_1$ [Eq.~\eqref{eq23}]. Note that the traceless nature of $\delta_1$, which we showed earlier, is particularly clear from this expression: $\sigma_1$ is a commutator [see Eq.~\eqref{eq23}], and hence has zero trace.

The second-order terms can be worked out similarly. We just provide the final expressions here: Eq.~\eqref{eq82} gives
\begin{align}
\,&\frac{1}{2}(\ddA)_{k\ell} \mu_k^4\frac{x_{\ell k}^4-1}{\ln(x_{\ell k})} \label{eq89}\\
=&\,\mu_k\mu_\ell (\sigma_2)_{k\ell} \nonumber\\
&-\frac{1}{4} \sum_m (\dA)_{km}(\dA)_{m\ell} \frac{\mu_k^4}{\ln(x_{\ell m})}\bigg(\frac{x_{\ell k}^4-1}{\ln(x_{\ell k})}-\frac{x_{mk}^4-1}{\ln(x_{mk})}\bigg), \nonumber
\end{align}
while Eq.~\eqref{eq78} yields 
\begin{align}
(\delta_2)_{k\ell}
=&\frac{1}{2}(\ddA)_{k\ell} \mu_k^2\frac{x_{\ell k}^2-1}{\ln(x_{\ell k})} \label{eq90}\\
&\hspace*{-0.5cm}+\frac{1}{4} \sum_m (\dA)_{km}(\dA)_{m\ell}\frac{ \mu_k^2}{\ln(x_{\ell m})}\bigg(\frac{x_{\ell k}^2-1}{\ln(x_{\ell k})}-\frac{x_{m k}^2-1}{\ln(x_{mk})}\bigg). \nonumber
\end{align}
$\ddA$ in Eq.~\eqref{eq89} can be shown to be Hermitian. Combining Eqs.~\eqref{eq86}, \eqref{eq89} and \eqref{eq90}, we get
\begin{align}\label{eq91}
(\delta_2)_{k\ell} =&\frac{\mu_k\mu_\ell}{s_{k\ell}} (\sigma_2)_{k\ell} -\frac{\mu_k\mu_\ell}{s_{k\ell}} \sum_m \frac{\mu_m^2}{s_{km}s_{\ell m}}(\sigma_1)_{km}(\sigma_1)_{m\ell},
\end{align}
indeed satisfied by our expressions for $\delta_2$ [Eqs.~\eqref{eq:V} and \eqref{eq:W}], $\sigma_1$, and $\sigma_2$ [Eq.~\eqref{eq23}].
In fact, from our analysis here, we can identify
\begin{align}\label{eq:VW}
(V)_{k\ell}&=-\frac{\mu_k\mu_\ell}{s_{k\ell}} \sum_m \frac{\mu_m^2}{s_{km}s_{\ell m}}(\sigma_1)_{km}(\sigma_1)_{m\ell},\nonumber\\
\textrm{and }\quad
(W)_{k\ell}&=\frac{\mu_k\mu_\ell}{s_{k\ell}} (\sigma_2)_{k\ell}.
\end{align}

In summary, we have shown that our \emph{ansatz} for $M^{1/2}$ in Eq.~\eqref{eq25} with our earlier expressions for $\delta_1$ and $\delta_2$ gives the correct (i.e., nonnegative) square root of $M$. In addition, we found alternate expressions that relate $\delta_1$ and $\delta_2$ to $\sigma_1$ and $\sigma_2$ [Eqs.~\eqref{eq88} and \eqref{eq:VW}].

\bibliographystyle{aps}

\begin{thebibliography}{10}
\bibitem{note1}
Ref.~\cite{Shabani+Lidar:05} discussed a Hermitian map, a system-only map defined for such an open-system situation. However, their map gives the correct description only for those input states with fixed system-bath correlations, a constraint difficult to ensure or verify in practice.
\bibitem{note2}
Note that, although every $F_{ak\ell}$ operator above can always be written as $\bra{k}F_a\ket{\ell}$ for some system-bath operator $F_a$, the set $\{F_a\}$ is generally not a Kraus representation for $\cE$.
\end{thebibliography}

\end{document}